\documentclass[12pt,a4paper]{article}

\usepackage[english]{babel}

\usepackage[letterpaper,top=2cm,bottom=2cm,left=3cm,right=3cm,marginparwidth=1.75cm]{geometry}

\usepackage[colorlinks=true, allcolors=blue]{hyperref}
\usepackage{amsmath,amssymb,amsthm}
\usepackage[numbers,sort&compress]{natbib}
\usepackage{tikz}
\usepackage{fullpage}

\usetikzlibrary{
        arrows.meta,
        decorations.markings,
        }

\usepackage{inputenc}

\newtheorem{Theorem}{Theorem}[section]

\newtheorem{Lemma}[Theorem]{Lemma}

\theoremstyle{definition}
\newtheorem{Definition}[Theorem]{Definition}
\newtheorem{Observation}[Theorem]{Observation}

\title{Classes of Phylogenetic Networks that are Robustly Orientable}
\author{Thomas Britz, Michael Hendriksen, Kaitlyn O'Reilly}

\tikzset{
    midarrow/.style={
        decoration={
            markings,
            mark=at position 0.5 with {\arrow{Latex[scale=1.2]}}
        },
        postaction={decorate}
    }
}

\begin{document}
\maketitle

\begin{abstract}
Standard phylogenetic reconstruction techniques often yield unrooted phylogenetic networks;
these are subsequently rooted to infer evolutionary history. 
A common problem in this process is to determine the structural classes to which the resulting network will belong. 
In this paper, we investigate unrooted networks in which the choice of any root results in 
a valid rooted phylogenetic network, a property we define as {\em robustly orientable}. 
We then establish a strict structural condition for this class, specifically, 
that an unrooted network is robustly orientable if and only if it contains no sink components. 
We also show that if an unrooted network is level-$2$ or less, or if it is tree-based, 
then it is robustly orientable. 
Furthermore, we define an unrooted network to be {\em robustly class $\mathcal C$} 
if the choice of any root results in a network belonging to class $\mathcal C$. 
We demonstrate that an unrooted network is robustly tree-child or robustly stack-free 
if and only if it is level-$1$ or less. 
Finally, we show that a phylogenetic network is robustly normal if and only if it is a phylogenetic tree. 
\end{abstract}

\section{Introduction}

Standard phylogenetic reconstruction techniques often result in an unrooted phylogenetic tree or network, 
which is then rooted through any number of techniques, e.g., 
through the use of outgroups or molecular clocks~\cite{felsenstein2003inferring}. 
Since the choice of root is not uniquely determined, there may be multiple plausible candidate root locations. 
It is therefore natural to ask whether 
desirable structural properties of the resulting rooted network depend on the choice of the root. 
In particular, once such a phylogenetic network is rooted, 
it is generally of interest as to which class of phylogenetic network it belongs. 
For instance, tree-child networks~\cite{cardona2009comparison} are both mathematically tractable 
(e.g., the Tree Containment Problem is solvable in linear time on tree-child networks~\cite{janssen2020linear} 
and tree-child networks are identifiable under certain models~\cite{francis2018identifiability}) 
and well-motivated biologically, as every internal vertex has a tree path to a leaf, 
which is a reasonable property to expect in many realistic biological contexts~\cite{kong2022classes}. 

The question of whether an unrooted phylogenetic network, once rooted, belongs to a particular class, 
has received a considerable amount of recent attention. 
For instance, Huber et al.~\cite{huber2024orienting} show that 
determining whether an unrooted phylogenetic network has a choice of root 
for which the resulting rooted network belongs to some class $\mathcal{C}$ 
is fixed-parameter tractable for classes that meet certain technical conditions. 
Bulteau et al.~\cite{bulteau2023turning} characterise when an unrooted phylogenetic network can 
even be oriented as a rooted phylogenetic network at all, 
and they provide an algorithm for doing so. 
In particular, determining whether an unrooted phylogenetic network can be oriented as 
a rooted tree-child network has received additional attention. 
For instance, Maeda et al.~\cite{maeda2023orienting} provide some necessary conditions, 
such as a tight upper bound on the size of tree-child orientable graphs, 
and show that any unrooted phylogenetic network can be made into one 
that can be oriented as a tree-child network through the addition of extra leaves. 
Furthermore, van Iersel et al.~\cite{van2026class} show that recognising whether 
an unrooted phylogenetic network admits a rooted tree-child orientation is NP-hard. 

In the present manuscript, 
we consider when an unrooted phylogenetic network can robustly be oriented to be a member of a particular class, 
by which we mean that any choice of root will result in a rooted phylogenetic network that is 
a member of that class. 
We call an orientation of an unrooted network \emph{phylogenetic}, 
if it meets the requirements to be a rooted phylogenetic network 
- a single root, at least one leaf, the only vertices of degree $0$ are leaves, and it is acyclic. 
If, for some given network, there exists some possible phylogenetic orientation for each choice of root location, 
then we call the network \emph{robustly orientable}.
Additionally, we call an unrooted robustly orientable phylogenetic network \emph{robustly} $\mathcal{C}$ if 
any phylogenetic orientation on this network is 
a member of some class $\mathcal{C}$ for any choice of root location. 

We first characterise those networks that are robustly orientable 
as precisely those networks that do not contain sink components. 
This can be viewed as a robust analogue to a result of Bulteau et al.~\cite{bulteau2023turning}, 
in which they showed that there must exist at least one phylogenetic orientation 
if there are fewer than two sink components. 
We also demonstrate that robust orientability is satisfied by several familiar network classes, 
in particular tree-based networks and level-$k$ networks for $k \le 2$, 
and we provide a network of minimum level and number of vertices that is not robustly orientable.

Interestingly, we also show that robustness collapses relatively complex properties in the rooted setting 
into simple ones in the unrooted setting. 
We show that an unrooted robustly orientable phylogenetic network is robustly tree-child if and only if 
it is robustly stack-free, and it is robustly stack-free if and only if it is level-$1$ or less. 
We additionally show that a phylogenetic network is robustly normal if and only if 
it is a phylogenetic tree. 
We end the manuscript with a discussion and possible future research directions.

\section{Preliminaries}

Throughout this manuscript we will consider both rooted and unrooted phylogenetic networks. 
All networks in this manuscript will be binary unless explicitly noted otherwise. 
We also note that all networks considered in this manuscript are unlabelled, 
as the results depend only on network \emph{shape}, 
but any results below will apply equally to labelled phylogenetic networks. 
We will begin by defining some familiar graph terminology, 
mostly following that of~\cite{berge2001theory}.

\begin{Definition}
    A {\em graph} $G$ is an ordered pair $(V(G),E(G))$ 
    where $V(G)$ is a set of elements, called the {\em vertices} of $G$, and 
    where $E(G)$ is a multiset of pairs of vertices, called the {\em edges} of~$G$.
    The graph $G$ is either {\em undirected}, 
    in which case each edge of $G$ is a set of two vertices, 
    or {\em directed}, in which case each edge is an ordered pair of vertices.
    The \emph{underlying undirected graph} of a directed graph $G$ is 
    the undirected graph obtained from $G$ by replacing each edge $(u,v)$ by $\{u,v\}$, 
    and reducing the multiplicity of any edge with multiplicity greater than one to a multiplicity of one. 
    A graph is called \emph{simple} if it contains no edges of multiplicity two or greater, 
    referred to as \emph{parallel} edges, 
    and if each edge contains two distinct vertices. 
    \end{Definition}

We now define paths and connectivity in graphs.

    \begin{Definition}
    A \emph{path} in a directed or undirected graph $G$ is a sequence $(v_1,e_1,v_2, \dots, e_{k-1},v_k)$ 
    that alternates between vertices and edges, 
    such that $e_i=\{v_i,v_{i+1}\}$ in the undirected case 
          and $e_i= (v_i,v_{i+1} )$ in the   directed case.
    We will occasionally refer to the latter as a \emph{directed path}. 
    A path is \emph{simple} if no vertex or edge appears more than once in the path.
    A path is a \emph{cycle} if $v_1=v_k$, 
    and a cycle $(v_1,e_1,v_2, \dots,v_{k-1}, e_{k-1},v_1)$ is \emph{simple} 
    if $(v_1,\dots,v_{k-1})$ is a simple path.
    A graph is \emph{connected} if there exists a path from $u$ to $v$ for any pair of vertices $u,v$ in the graph.
    A directed graph $G$ is \emph{weakly connected} if the underlying undirected graph of $G$ is connected. 
    An edge $e=(u,v)$ in a directed graph is referred to as a \emph{shortcut} 
    if there exists a directed path from $u$ to $v$ that does not contain $e$.
\end{Definition}

We can now define the main objects of study: 
rooted and unrooted phylogenetic networks.

\begin{Definition}
	An {\em unrooted phylogenetic network} is a finite, simple, connected, undirected graph $G=(V,E)$ with at least one vertex of degree $1$. 
    The vertices of degree $1$ are referred to as {\em leaves}.
    Each non-leaf vertex is an {\em internal vertex}.
    If each internal vertex has degree $3$, then the network is referred to as \emph{binary}.
\end{Definition}

\begin{Definition}
    A {\em rooted phylogenetic network} is 
    a finite, simple, weakly connected, directed, acyclic graph $T=(V,E)$ with 
    a unique vertex of in-degree $0$ and out-degree $2$, 
    referred to as the \emph{root} and denoted $\rho$, 
    with at least one vertex of degree $1$, 
    where the vertices of degree~$1$ are referred to as {\em leaves} and are the only vertices with out-degree $0$. 
    Each non-leaf, non-root vertex is an {\em internal vertex}.
    If all vertices except the root have degree $3$, then the network is referred to as \emph{binary}.
\end{Definition}

We also require some language to describe the edges and vertices in these networks.

\begin{Definition}
    If a vertex in a rooted phylogenetic network has in-degree less than or equal to $1$, 
    then it is referred to as a \emph{tree} vertex; 
    otherwise, it is referred to as a \emph{reticulation} vertex. 
    An incoming edge to a reticulation vertex is referred to as a \emph{reticulation edge}; 
    otherwise, it is a \emph{tree edge}.
    If there is a directed edge from vertex~$u$ to vertex~$v$, 
    then $u$ is referred to as a \emph{parent} of~$v$, 
    and  $v$ is referred to as a \emph{child}  of~$u$. 
\end{Definition}

\begin{Definition}
    A rooted phylogenetic network is referred to as \emph{tree-child}~\cite{cardona2009comparison} 
    if every internal vertex has at least one child that is a tree vertex.
    A rooted phylogenetic network is referred to as \emph{stack-free} 
    if every reticulation vertex only has tree children.
\end{Definition}

Another common class of phylogenetic networks is the class of normal networks, 
introduced by Willson~\cite{willson2010properties}.

\begin{Definition}
    A rooted phylogenetic network is referred to as \emph{normal} if 
    it is tree-child and has no shortcuts.
\end{Definition}

\begin{Definition}
    Let $G$ be an undirected graph, and let $e \in E(G)$. 
    An \emph{orientation} of $(G,e)$, denoted $o(G,e)$, 
    is any directed graph obtained from $G$ by subdividing $e$ to form a degree-$2$ vertex, 
    and then replacing each undirected edge $\{u,v\}$ by a directed edge, 
    either $(u,v)$ or $(v,u)$.
\end{Definition}

\begin{Definition}
    An orientation $o(G,e)$ is \emph{phylogenetic}  
    if it has exactly one root; 
       it is acyclic; 
       it has at least one leaf; and
       its leaves are the only vertices with out-degree~$0$.
\end{Definition}

Observe that any phylogenetic orientation of a graph is a rooted phylogenetic network; 
however, the converse is not necessarily true, as per the following observation.

\begin{Observation}
    Not all rooted phylogenetic networks can come from orienting an unrooted one. 
    For example, the two-leaf rooted tree with a single reticulation edge added, 
    depicted in Figure~\ref{fig:two-leaves}, 
    when the root is suppressed and edges made undirected, 
    becomes an unrooted network with two parallel edges. 
    This would not be robustly orientable even if parallel edges were permitted, 
    since only one of the parallel edges could be chosen to be the root.
\end{Observation}

\begin{figure}[ht]
    \centering
\begin{tikzpicture}
[
    scale=0.6,
    dot/.style={circle, fill=black, inner sep=1.67pt}
]

% Root (no dot)
\coordinate (root) at (3,4);

% Internal vertices
\node[dot] (a) at (1.5,2) {};
\node[dot] (b) at ( 5,1) {};

% Leaves (no dots)
\coordinate (l1) at (0,0);
\coordinate (l2) at ( 5.6,0);

\node at (2.8,-1) {(a)};

% Tree edges
\draw[midarrow, thick] (root) -- (a);
\draw[midarrow, thick] (root) -- (b);
\draw[midarrow, thick] (a) -- (l1);
\draw[midarrow, thick] (b) -- (l2);

% Reticulation edge
\draw[midarrow, thick] (a) -- (b);

\end{tikzpicture}
\hspace{2cm}
\begin{tikzpicture}
[
    scale=0.6,
    dot/.style={circle, fill=black, inner sep=1.67pt}
]

    \coordinate (a) at (0,0) {};
    \node[dot] (b) at (3,0) {};
    \node[dot] (c) at (6,0) {};
    \coordinate (d) at (9,0) {};

    \node at (4.5,-2) {(b)};

    \draw (a) -- (b);
    \draw (c) -- (d);
    \draw    (b) to[out=45,in=135] (c);
    \draw    (b) to[out=-45,in=-135] (c);
\end{tikzpicture}
    \caption{(a) An example of a rooted phylogenetic network $N$ that 
        cannot be obtained by orienting an unrooted phylogenetic network. 
    (b) The network $N$ with orientations removed, resulting in a graph with a pair of parallel edges, 
    which is not a simple graph and thus not a phylogenetic network.}
    \label{fig:two-leaves}
\end{figure}
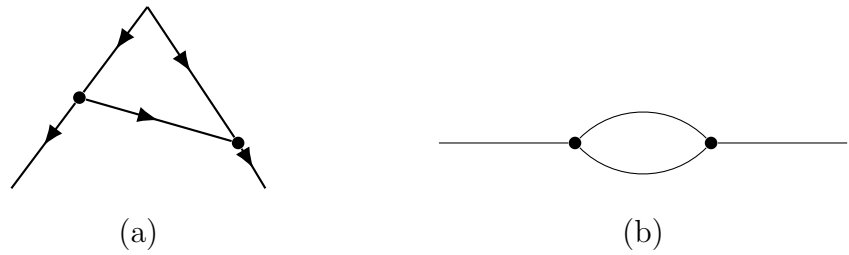

We note the following obvious but useful fact about phylogenetic orientations.

\begin{Lemma}\label{l:oneroot}
    Let $G$ be an acyclic directed graph for which the underlying undirected graph is connected, 
    with precisely one vertex $\rho$ of in-degree~$0$. 
    Then every vertex $v$ in $G$ has a directed path from $\rho$ to~$v$.
\end{Lemma}

We can now define the main objects of study for the present manuscript. 
We first consider undirected graphs that are robustly orientable, 
which means that they are able to produce a phylogenetic network for any choice of root.

\begin{Definition}
    An undirected graph $G$ is \emph{robustly orientable} if 
    there exists a phylogenetic orientation $o(G,e)$ for each edge $e$ in~$G$.
\end{Definition}

We will also consider unrooted networks that can robustly be oriented to produce 
members of a class of rooted networks.

\begin{Definition}
    Let $N$ be an unrooted network, let $e \in E(N)$ 
    and let $\mathcal{C}$ be a class of rooted phylogenetic networks. 
    Then the tuple $(N,e)$ is referred to as \emph{robustly $\mathcal{C}$} if 
    every phylogenetic orientation of $(N,e)$ is a rooted phylogenetic network in~$\mathcal{C}$. 
    Furthermore, $N$ is referred to as \emph{robustly $\mathcal{C}$} if, 
    for all $e \in E(N)$, $(N,e)$ is robustly~$\mathcal{C}$. 
\end{Definition}

With the above definition, we can refer to an unrooted network as, for example, 
robustly tree-child or robustly stack-free without ambiguity. 
We note in particular that, 
while robust orientability requires 
the existence of \emph{some} phylogenetic orientation for each choice of root, 
robust $\mathcal{C}$-ness of a network requires 
\emph{every} phylogenetic orientation for \emph{every} choice of root to be a member of class $\mathcal{C}$.

\begin{Definition}
    Let $G$ be a connected directed or undirected graph. 
    A \emph{cut-edge} of $G$ is any edge $e\in E(N)$ such that 
    $G \backslash \{e\}$ is not a connected graph, when considering the underlying undirected graph. 
\end{Definition}

\begin{Definition}\label{d:oriented_away}
Suppose that $G$ is a connected undirected graph with a cut-edge $c$, 
and let $u$ be a vertex of~$G$.
Let $C$ be the connected component of $G \backslash \{c\}$ that contains $u$.
Let $v$ be the vertex of $c$ that is in $C$ and let $w$ be the other vertex of $c$.
Consider any orientation $o(G,e)$.
If $c$ is oriented as $(v,w)$ in $o(G,e)$, 
then $c$ is said to be \emph{oriented away from}~$u$ in $o(G,e)$.
\end{Definition}

\begin{Theorem}\label{t:oriented_away}
    Let $o(N,e)$ be a phylogenetic orientation with root $\rho$. 
    Then every cut-edge in $N$ is oriented away from $\rho$ in $o(N,e)$.
\end{Theorem}

\begin{proof}
    Let $c$ be a cut-edge in $N$ and consider the two components obtained by deleting $c$. 
    If $c$ is not oriented away from $\rho$, 
    then no vertex $v$ in the component not containing $\rho$ 
    has a directed path from $\rho$ to $v$, 
    contradicting Lemma~\ref{l:oneroot}.
\end{proof}

\begin{Definition}
    A \emph{biconnected component}, or \emph{blob}, 
    of a directed or undirected graph $G$ is a maximal subgraph $H$ of $G$ 
    that is not a single vertex, such that $H \backslash \{e\}$ is (weakly) connected for any edge $e \in E(H)$.
\end{Definition}

\begin{Definition}
    Let $B$ be a blob of a directed or undirected graph $G$, 
    and let $C(B)$ be the set of cut-edges $\{u,v\}$ such that 
    either $u$ or $v$ is contained in $B$. 
    The \emph{decorated blob of $B$} is the subgraph of $G$ consisting of $B$ and $C(B)$,
    that is, the union of $B$ with all of its adjacent cut-edges. 
    We will refer to any graph as a decorated blob if it can be expressed as such a union. 
\end{Definition}

\begin{Definition}
    A \emph{sink component} is a biconnected component that is incident with exactly one cut-edge.
\end{Definition}

Given any biconnected undirected graph~$G$ and any edge $e=\{s,t\}$, 
the graph $G\setminus\{e\}$  can be oriented, 
for instance using {\em $st$-numbering}, 
to give an acyclic directed graph in which $s$ has in-degree~$0$ and $t$ has out-degree~$0$.
In fact, the following stronger statement is true;
see~\cite{LempelEvenCederbaum1966,even1976computing}.
  
\begin{Theorem} \label{ro:st_orientation}
    Let $G$ be a biconnected undirected graph. 
    Then, for each edge $e=\{s,t\}$ of~$G$, 
    there is an orientation of $G\setminus\{e\}$ that is an acyclic directed graph, where
    $s$ has in-degree~$0$, $t$ has out-degree~$0$, 
    and every other vertex has positive in- and out-degree.
\end{Theorem}

\section{Robustly Orientable Networks}

We first consider the class of robustly orientable networks - 
that is, the class of unrooted networks $N$ for which, for any choice of edge subdivided to make a root, 
there exists an orientation on $N$ that produces a directed phylogenetic network.

We note that Bulteau et al.~\cite{bulteau2023turning} consider the analogous problem for existence - 
that is, they characterise the class of unrooted networks for which 
there exists \emph{some} edge that can be subdivided to form a root, 
and then an orientation is imposed to form a directed phylogenetic network. 
By comparison, our class requires every edge to have this property. 
Their characterisation is that an unrooted network admits such an orientation 
if there are fewer than two sink components, 
while we will show that robust orientability requires the unrooted network to have no sink components. 
In this way, the present result can be considered to be a robust analogue to the work of Bulteau et al. 
Our proof is constructive, 
in the sense that we describe how to construct a valid phylogenetic orientation on such a network, 
using the aforementioned $st$-numbering.

\begin{Theorem}\label{t:no.sinks}
    An unrooted phylogenetic network $N$ is robustly orientable 
    if and only if it contains no sink component.
\end{Theorem}

\begin{proof}
    We first show that if $N$ has a sink component, 
    then it cannot be robustly orientable. 
    In particular, suppose that $N$ contains a sink component $S$, 
        let $e$ be the unique cut-edge incident with~$S$, 
    and let $v$ be the vertex of~$e$ not in~$S$.

    Assume that there exists some phylogenetic orientation $o = o(N,e)$, 
    and let $\rho$ be the root, obtained from subdividing~$e$, in that orientation.
    Let $o'$ be the subdigraph of $o$ whose underlying graph is~$S$. 
    Consider a maximal directed path $\rho,v,\ldots$ in $o$; 
    all vertices in this path, except for $\rho$, lie in $o'$. 
    Since $o$, and thus $o'$, is finite and acyclic, 
    the path must end in some vertex $z$ of $o'$ with out-degree~$0$.
    However, $S$ is a bicomponent, so $z$ is not a leaf.
    Therefore, $o$ is not a phylogenetic network, a contradiction.
    It follows that, if $N$ contains a sink component, then $N$ is not robustly orientable. 
    
    We now show that, if $N$ has no sink component, then it is robustly orientable.
    Select any edge $e =\{v,w\}\in E(N)$,
    subdivide it to form the vertex $\rho$, 
    and let $N^e$ be the resulting unrooted network. 
    We construct an orientation on $N^e$ as follows.

    If $e$ is a cut-edge of $N$, 
    then orient the adjacent edges of $\rho$ as $(\rho,v)$ and $(\rho,w)$. 
    Otherwise, $e$ is contained in a biconnected component~$S$.
    Since $N$ is a phylogenetic network, it contains at least one leaf, so $S\neq N$. 
    Since $N$ is connected, 
    $S$ has a vertex $t$ that is contained in some edge $f$ not in $S$.
    Since the vertices of $N$ have degree at most three, 
    no two components of $N$ can share a vertex, 
    so $f$ cannot belong to a component of $N$; 
    in other words, $f$ is a cut-edge.
    By Theorem~\ref{ro:st_orientation}, 
    there exists an acyclic orientation $o(S,e)$ in which 
    $\rho$ has in-degree $0$, 
    $t$ has out-degree~$0$, 
    and every other vertex has positive in- and out-degree.

    Now consider any biconnected component $S'$ not containing~$e$. 
    Since $N$ is connected, there is a path from $e$ to some vertex of $S'$;
    let $s'$ be the first vertex along that path which is contained in~$S'$. 
    Since $N$ has no sink component, 
    there must be at least one other vertex $t'$ in $S'$ that is incident to an edge not in $S'$.
    By Theorem~\ref{ro:st_orientation}, 
    there is an acyclic orientation on $S'$ such that $s'$ has in-degree $0$ in $S'$, 
    $t'$ has out-degree~$0$, 
    and every other vertex has positive in- and out-degree.
    Find such acyclic orientations for each biconnected component $S'$ in $N$.

    Note that there is no other vertex $s''$ in $S'$ that is first in $S'$ along some path from $\rho$
    since the union of these two paths with $S'$ would form 
    a biconnected component that strictly contains~$S'$, 
    a contradiction. 
    The orientation on $S'$ is therefore well-defined.

    Finally, orient each cut-edge in $N$ away from the root $\rho$, 
    as in Definition~\ref{d:oriented_away}.

    We claim that this yields a phylogenetic orientation $o$ of $N^e$.
    In particular, 
    we have already shown that the orientation $o$ is well-defined;
    also, $N$ contains at least one leaf, so $o$ does too. 
    By construction, each vertex except $\rho$ has positive in-degree, 
    so $o$ has the unique root~$\rho$;
    similarly, only the leaves of $N$ have out-degree~$0$. 
    
    Finally, the orientation $o$ is acyclic 
    since any cycle would be contained in some biconnected component, 
    a contradiction since each biconnected component has acyclic orientation.
    
    We have thus constructed a phylogenetic orientation,
    and the theorem is proved.
\end{proof}

\begin{figure}[htbp]
   \centering
   \begin{tikzpicture}
   [
   scale=0.6,
   dot/.style={circle, fill=black, inner sep=1.67pt}
   ]

   %Define extra leaf and where root goes to make contradiction
   \node[dot] (leaf) at (0,3) {};
   \node[dot] (connector) at (2,3) {};

   %Define corners where UL is upper left 
   \node[dot] (UL) at (2,5) {};
   \node[dot] (LL) at (2,1) {};
   \node[dot] (UR) at (6,5) {};
   \node[dot] (LR) at (6,1) {};

   %Draw edges 
   \draw (connector) -- (UL);
   \draw (connector) -- (LL);
   \draw (leaf) -- (connector);
   \draw (UL) -- (UR);
   \draw (UR) -- (LR);
   \draw (LR) -- (LL);
   \draw (LR) -- (UL);
   \draw (UR) -- (LL);
\end{tikzpicture}
   \caption{An unrooted network containing a sink component}
   \label{fig: Simple Sink Component}
\end{figure}
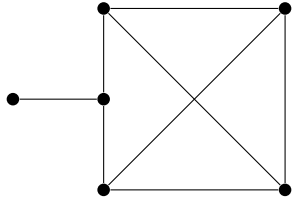

We provide an example of an unrooted binary network with a sink component in Figure~\ref{fig: Simple Sink Component};
this is the smallest example, in the sense of fewest vertices. 
However, there are several familiar classes that do not permit sink components, 
and we consider two of them now - tree-based networks and level-$k$ networks, for $k \le 2$.

\begin{Definition}
    An unrooted phylogenetic network $N$ is {\em tree-based}~\cite{francis2015phylogenetic} exactly when 
    it has a spanning tree whose leaf set coincides with the leaf set of~$N$.
\end{Definition}

\begin{Theorem}
    Every tree-based unrooted phylogenetic network is robustly orientable. 
\end{Theorem}

\begin{proof}
    Let $N$ be a tree-based unrooted phylogenetic network and 
    let $T$ be a spanning tree of $N$ whose leaves are those in $N$. 
    
    Assume that $N$ contains some sink component $S$, 
    and let $e$ be the cut-edge incident with~$S$.
    Let $T'$ be the subtree of $T$ restricted to $S$; 
    it is a spanning tree of $S$.
    Since $T'$ is a tree with more than one vertex, 
    it has at least two leaves.
    The cut-edge $e$ is incident with one of these, 
    so some leaf $u$ of $T'$ is a vertex of~$S$. 
    Since $N$ is tree-based, $u$ is also a leaf of~$S$, 
    a contradiction since $S$ contains no leaf. 
    
    Therefore, $N$ contains no sink component.
    By Theorem~\ref{t:no.sinks}, $N$ is robustly orientable. 
\end{proof}

\begin{Definition}
     An unrooted phylogenetic network is \emph{level-$k$}~\cite{choy2005computing} if, 
     by deleting at most $k$ edges from each biconnected component, the resulting graph is a tree, 
     that is, it has no simple cycle, 
     and at least one biconnected component requires the deletion of $k$ edges to become a tree. 
     A directed phylogenetic network is \emph{level-$k$} if 
     its underlying unrooted phylogenetic network is level-$k$. 
\end{Definition}

Note that a directed binary phylogenetic network $N$ is level-$k$ if and only if 
the maximal number of reticulations in the biconnected components $S$ of~$N$ is~$k$.
Equivalently, the \emph{nullity} or \emph{cyclomatic number}~\cite{berge2001theory}
$n(S)=|E(S)|-|V(S)|+1$ of each biconnected component $S$ is at most~$k$, 
with at least one component having nullity~$n(S)=k$. 
For example, the network depicted in Figure~\ref{fig: Simple Sink Component} is level-$3$.

\begin{Theorem}\label{t:level2robust}
    For $k\leq 2$, every level-$k$ unrooted binary phylogenetic network $N$ is robustly orientable. 
\end{Theorem}

\begin{proof}
    For some $k\leq 2$, let $N$ be a level-$k$ unrooted binary phylogenetic network.
    If $k=0$, then $N$ has no biconnected component, 
    so $N$ is robustly orientable by Theorem~\ref{t:no.sinks}.

    Assume that $N$ contains a sink component~$S$ with incident cut-edge~$e$, 
    and let $u$ be the vertex in $S$ that is incident with~$e$.
    In $S$, each vertex has degree $3$ in~$S$, except $u$ which has degree~$2$. 
    Note that $|V(S)|\geq 4$
    and that, by the well-known handshaking lemma, 
    $2|E(S)| = 3|V(S)|-1$.
    Therefore, the level of $S$ is 
    \[
      k = |E(S)| - |V(S)| + 1 
        = \frac{1}{2}\big(3|V(S)| - 1\big) - |V(S)| + 1 
        = \frac{1}{2}\big( |V(S)| + 1\big)
      \geq\frac{5}{2}\,.
    \]
    It follows that $k\geq 3$, 
    a contradiction.
    Therefore, 
    $N$ has no sink component, 
    so $N$ is robustly orientable by Theorem~\ref{t:no.sinks}. 
\end{proof}

Theorem~\ref{t:level2robust} is best possible, 
in the sense that the bound $k\leq 2$ is tight. 
This is demonstrated by the network depicted in Figure~\ref{fig: Simple Sink Component}, 
which is level-$3$ but is not robustly orientable since it contains a sink component.

Finally, we characterise the structure of level-$1$ unrooted networks. 
It is well-known that an unrooted phylogenetic network is level-$k$ for $k \le 1$ if and only if 
every non-trivial biconnected component is a simple cycle (see e.g.,~\cite[Observation 16]{hellmuth2023clustering}); 
we state the following equivalent result.

\begin{Theorem}\label{t:level1.structure}
    An unrooted binary phylogenetic network $N$ is level-$k$ for $k \le 1$ if and only if 
    every decorated blob of $N$ consists of a simple cycle $C$ 
    together with a single leaf affixed to each vertex of~$C$.
\end{Theorem}

\begin{Observation}\label{o:level1.implies.stack-free.tree-child}
  Note that each level-$k$ rooted binary phylogenetic network for $k \le 1$ is stack-free and tree-child, 
  since each biconnected component contains a single reticulation vertex; 
  see also~\cite{kong2022classes}.
\end{Observation}

\section{Robustly tree-child and robustly stack-free networks}

We now consider robustly tree-child and robustly stack-free networks. 
We will first show that all level-$k$ unrooted networks for $k \le 1$ are in these classes; 
later, we will show that these three classes in fact coincide.

\begin{Theorem}\label{t:yes.level1}
    Every level-$k$ robustly orientable network $N$ for $k \le 1$
    is robustly tree-child and robustly stack-free.
\end{Theorem}

\begin{proof}
    As $N$ is robustly orientable, any choice of root admits a phylogenetic orientation. 
    By Observation~\ref{o:level1.implies.stack-free.tree-child}, 
    the resulting network is stack-free and tree-child.
\end{proof}

Our general strategy below will be to adjust an unrooted network slightly by performing surgery on an area 
so that we can then orient it so as to create a forbidden substructure, 
and then transform the now rooted network back in a way that 
preserves the phylogenetic orientation and has the original network as the underlying undirected network. 
In order to construct such a forbidden substructure for the present classes, 
we first identify a type of vertex that will be useful.

\begin{Definition}
    A vertex with no incident cut-edge is referred to as \emph{lonely}. 
\end{Definition}

We also need a simple lemma to aid us in a counting argument for the main proofs.

\begin{Lemma}\label{l:level.B.D.equal}
    Let $N$ be an unrooted phylogenetic network, $B$ a biconnected component of $N$, 
    and $D$ the decorated blob obtained by taking the union of $B$ with its incident cut-edges. 
    Then the level of $B$ is equal to the level of $D$.
\end{Lemma}

\begin{proof}
    It suffices to observe that, 
    if we adjoin $m$ cut-edges, then $m$ edges and $m$ vertices are added, 
    and so $|E(D)|-|V(D)|+1=|E(B)|+m-(|V(B)|+m)+1=|E(B)|-|V(B)|+1$, 
    leaving the level unchanged.
\end{proof}

We will now show that a sufficiently complex unrooted network must contain a lonely vertex, 
which we will later modify in order to create our forbidden substructures.

\begin{Lemma}\label{l:lonely.friends}
    Let $N$ be a level-$k$ binary unrooted phylogenetic network for $k \ge 2$. 
    Then there exists a lonely vertex in $N$ that is adjacent to a non-leaf vertex that is not lonely.
\end{Lemma}

\begin{proof}
    As $N$ is level-$k$, there exists a biconnected component $B$ in $N$ for which $k=|E(B)|-|V(B)|+1$. 
    We first show that there exists at least one lonely vertex in $B$ and thus in $N$. 
    By Lemma~\ref{l:level.B.D.equal}, 
    the decorated blob $D$ obtained by adjoining the incident cut-edges of $B$ is also level-$k$. 
    
    Assume that there are no lonely vertices in $D$. 
    Then each vertex $v$ in $B$ is incident with at least one cut-edge,
    As $B$ is a biconnected component, it contains no cut-edge, 
    so $v$ is incident with exactly one cut-edge. 
    This follows from the fact that $N$ is binary, 
    so if $v$ were adjacent to $3$ cut-edges, 
    then $N$ would be a single vertex attached to $3$ cut-edges and so would be level-$0$, 
    and if it were attached to $2$ cut-edges, then its third edge would also be a cut-edge, 
    contradicting the fact that $B$ is a biconnected component. 
    
    Then $|I|=|L|$ where $I$ is the set of vertices of $B$ 
    and where $L$ is the set of vertices in the cut-edges that are not contained in $B$ by $L$.
    Therefore, $|V(D)| = |I|+|L| = 2|I|$. 
    In~$D$, vertices in $I$ are of degree $3$ and vertices in $L$ are of degree $1$, 
    so by the Handshaking Lemma,
    $|E(D)| = (3|I|+|L|)/2 = 4|I|/2 = 2|I| = |V(D)|$, so $k=|E(D)|-|V(D)|-1=1$, 
    which contradicts the fact that $k \ge 2$. 
    It follows that there is at least one lonely vertex.

    Finally, as $N$ is a phylogenetic network, it also contains at least one leaf, 
    and therefore, since $N$ is connected and has more than one edge, 
    a non-leaf vertex that is not lonely. 
    Since $N$ is connected, 
    some non-lonely internal vertex must be adjacent to one of the lonely vertices.
\end{proof}

We can now prove a lemma that allows us to perform surgery on part of a valid orientation 
without disrupting its validity.

\begin{Lemma} \label{l:leaf.fusion}
    Let $N$ be a robustly orientable phylogenetic network with at least $2$ leaves, 
    and let $o(N,e)$ be a phylogenetic orientation on $N$. 
    Construct the network $N^*$ by identifying two leaves in $N$ 
    and adjoining a new leaf to the resultant degree-$2$ vertex $v$.
    Define the orientation $o^*(N^*,e)$ on $N^*$ to be the same as $o(N,e)$ except for the new leaf, 
    which is directed towards~$v$. 
    Then $N^*$ is an unrooted phylogenetic network and $o^*(N^*,e)$ is a phylogenetic orientation.
\end{Lemma}

\begin{proof}
    By definition, $N^*$ is a phylogenetic network, as it still has at least one leaf, 
    so it remains to show that $o^*(N^*,e)$ is a phylogenetic orientation.
    
    As we identified two leaves and $o$ was a phylogenetic orientation, 
    in $o^*(N,e)$, the two old edges still point towards the identified vertex, 
    and the third edge points out to the leaf. 
    Now, $o^*(N^*,e)$ still has at least one leaf, 
    the new leaf is a sink and so the sinks are still precisely the leaves, 
    and no new root was added. 
    Assume that $o^*(N^*,e)$ contains some directed cycle.  
    Since $o(N,e)$ is acyclic, 
    the directed cycle must contain any of the two former leaf edges, or the new leaf edge. 
    But any path that uses any of those three edges must end either at the identified vertex or the new leaf, 
    and hence cannot be a cycle, a contradiction.
    Therefore, $o^*(N^*,e)$ is acyclic, and we have proved that it is a phylogenetic orientation.
\end{proof}

Finally, we can now show that unrooted networks with level at least $2$ can always 
admit a valid phylogenetic orientation that is not tree-child or stack-free.

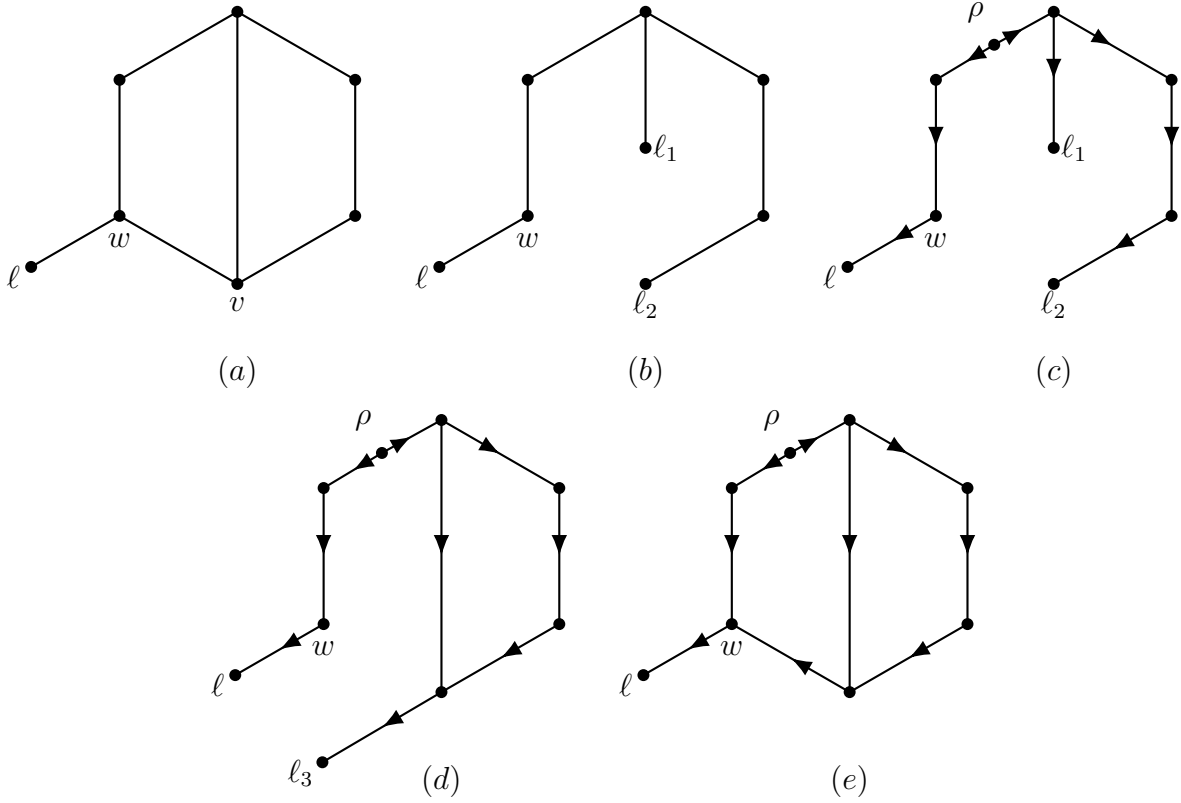
\begin{figure}[ht]
    \centering
\begin{tikzpicture}[scale=0.9]
  \begin{scope}[shift={(-6,0)}]
    %\draw[help lines] (-5,-5) grid (5,5);
    \foreach \angle [count = \i] in {30, 90, ..., 330} {
        \fill (\angle:2cm) circle (2.5pt) coordinate (dot\i);
        %\node at (\angle: 2.4) {Dot \i};
    }
    \node at (270:3.3cm) {$(a)$};

    \node at (270:2.3cm) {$v$};
    \node at (218:2.2cm) {$w$};
    \node at (210:3.8cm) {$\ell$};
    \fill (210:3.5cm) circle (2.5pt) coordinate (leaf1);

    \draw [thick]  (dot2) -- (dot1);
    \draw [thick]  (dot3) -- (dot2);
    \draw [thick]  (dot3) -- (dot4);
    \draw [thick]  (dot1) -- (dot6);
    \draw [thick]  (dot6) -- (dot5);
    \draw [thick]  (dot5) -- (dot4);
    \draw [thick]  (dot5) -- (dot2);
    \draw [thick]  (dot4) -- (leaf1);
\end{scope}
\begin{scope}[shift={(0,0)}]
    %\draw[help lines] (-5,-5) grid (5,5);
    \foreach \angle [count = \i] in {30, 90, ..., 330} {
        \fill (\angle:2cm) circle (2.5pt) coordinate (dot\i);
        %\node at (\angle: 2.4) {Dot \i};
    }
    \node at (270:3.3cm) {$(b)$};
    
    \node at (270:2.3cm) {$\ell_2$};
    \node at (218:2.2cm) {$w$};
    \node at (210:3.8cm) {$\ell$};
    \node at (0:0.3cm) {$\ell_1$};
    \fill (210:3.5cm) circle (2.5pt) coordinate (leaf1);
    \fill (0:0cm) circle (2.5pt) coordinate (leaf2);

    \draw [thick]  (dot2) -- (dot1);
    \draw [thick]  (dot3) -- (dot2);
    \draw [thick]  (dot3) -- (dot4);
    \draw [thick]  (dot1) -- (dot6);
    \draw [thick]  (dot6) -- (dot5);
    \draw [thick]  (leaf2) -- (dot2);
    \draw [thick]  (dot4) -- (leaf1);
\end{scope}
\begin{scope}[shift={(6,0)}]
    %\draw[help lines] (-5,-5) grid (5,5);
    \foreach \angle [count = \i] in {30, 90, ..., 330} {
        \fill (\angle:2cm) circle (2.5pt) coordinate (dot\i);
        %\node at (\angle: 2.4) {Dot \i};
    }
    \node at (270:3.3cm) {$(c)$};
    
    \node at (270:2.3cm) {$\ell_2$};
    \node at (218:2.2cm) {$w$};
    \node at (210:3.8cm) {$\ell$};
    \node at (0:0.3cm) {$\ell_1$};
    \node at (120:2.3cm) {$\rho$};
    \fill (120:1.75cm) circle (2.5pt) coordinate (root);
    \fill (210:3.5cm) circle (2.5pt) coordinate (leaf1);
    \fill (0:0cm) circle (2.5pt) coordinate (leaf2);

    \draw [midarrow, thick]  (dot2) -- (dot1);
    \draw [midarrow, thick]  (root) -- (dot2);
    \draw [midarrow, thick]  (root) -- (dot3);
    \draw [midarrow, thick]  (dot3) -- (dot4);
    \draw [midarrow, thick]  (dot1) -- (dot6);
    \draw [midarrow, thick]  (dot6) -- (dot5);
    \draw [midarrow, thick]  (dot2) -- (leaf2);
    \draw [midarrow, thick]  (dot4) -- (leaf1);
\end{scope}
\begin{scope}[shift={(-3,-6)}]
    %\draw[help lines] (-5,-5) grid (5,5);
    \foreach \angle [count = \i] in {30, 90, ..., 330} {
        \fill (\angle:2cm) circle (2.5pt) coordinate (dot\i);
        %\node at (\angle: 2.4) {Dot \i};
    }
    \node at (270:3.3cm) {$(d)$};
    
    \node at (218:2.2cm) {$w$};
    \node at (210:3.8cm) {$\ell$};
    \node at (120:2.3cm) {$\rho$};
    \node at (237:3.8cm) {$\ell_3$};
    \fill (120:1.75cm) circle (2.5pt) coordinate (root);
    \fill (210:3.5cm) circle (2.5pt) coordinate (leaf1);
    \fill (240:3.5cm) circle (2.5pt) coordinate (leaf3);

    \draw [midarrow, thick]  (dot2) -- (dot1);
    \draw [midarrow, thick]  (root) -- (dot2);
    \draw [midarrow, thick]  (root) -- (dot3);
    \draw [midarrow, thick]  (dot3) -- (dot4);
    \draw [midarrow, thick]  (dot1) -- (dot6);
    \draw [midarrow, thick]  (dot6) -- (dot5);
    \draw [midarrow, thick]  (dot2) -- (dot5);
    \draw [midarrow, thick]  (dot4) -- (leaf1);
    \draw [midarrow, thick]  (dot5) -- (leaf3);
\end{scope}
\begin{scope}[shift={(3,-6)}]
    %\draw[help lines] (-5,-5) grid (5,5);
    \foreach \angle [count = \i] in {30, 90, ..., 330} {
        \fill (\angle:2cm) circle (2.5pt) coordinate (dot\i);
        %\node at (\angle: 2.4) {Dot \i};
    }
    \node at (270:3.3cm) {$(e)$};
    
    \node at (218:2.2cm) {$w$};
    \node at (210:3.8cm) {$\ell$};
    \node at (120:2.3cm) {$\rho$};
    \fill (120:1.75cm) circle (2.5pt) coordinate (root);
    \fill (210:3.5cm) circle (2.5pt) coordinate (leaf1);

    \draw [midarrow, thick]  (dot2) -- (dot1);
    \draw [midarrow, thick]  (root) -- (dot2);
    \draw [midarrow, thick]  (root) -- (dot3);
    \draw [midarrow, thick]  (dot3) -- (dot4);
    \draw [midarrow, thick]  (dot1) -- (dot6);
    \draw [midarrow, thick]  (dot6) -- (dot5);
    \draw [midarrow, thick]  (dot2) -- (dot5);
    \draw [midarrow, thick]  (dot4) -- (leaf1);
    \draw [midarrow, thick]  (dot5) -- (dot4);
\end{scope}
\end{tikzpicture}
    
    \caption{An example of our construction to show that level-$k$ networks for $k \ge 2$ are not 
    robustly tree-child/stack-free in Theorem~\ref{t:not.level2}. 
    All diagrams are restricted to a single biconnected component and the relevant leaves. 
    (a) The network $N$ with lonely vertex $v$, non-lonely vertex $w$, and leaf $\ell$; 
    (b) The network $N^*$; 
    (c) The orientation $o(N^*,e)$;
    (d) The orientation $o^*(N^{**},e)$; 
    (e) The orientation $o^{**}(N,e)$, which is an orientation on the original network $N$.}
    \label{fig:tc.construction}
\end{figure}

\begin{Theorem}\label{t:not.level2}
    Let $N$ be a level-$k$ robustly orientable phylogenetic network for $k \ge 2$. 
    Then $N$ is not robustly tree-child or robustly stack-free.
\end{Theorem}

\begin{proof} 
    By Lemma~\ref{l:lonely.friends}, 
    there exists in $N$ a lonely vertex $v$ adjacent to a non-lonely internal vertex $w$.
    Let $\ell$ be the leaf adjacent to $w$, 
    and denote the component containing $v$ by~$B$.

    We will perform a series of operations to modify $N$ and give it an orientation, 
    and then modify the underlying network back to~$N$. 
    We depict the operations in Figure~\ref{fig:tc.construction}, 
    and an example network $N$ is depicted in Figure~\ref{fig:tc.construction} (a).   

    Construct a network $N^*$ by deleting $v$ and all incident edges 
    and then adjoining leaves $\ell_1$ and $\ell_2$ to the vertices $x,y\neq w$ 
    that were adjacent to $v$. 
    This is depicted in Figure~\ref{fig:tc.construction} (b).

    We claim that $N^*$ is also robustly orientable; 
    that is, by Theorem~\ref{t:no.sinks}, that it contains no sink component. 
    Let $N^*|_B$ denote the subgraph of $N^*$ induced by the vertices in $B\setminus\{v\}$. 
    Observe that no component outside of $N^*|_B$ is a sink component, 
    as we have not deleted or added any cut-edge outside $N^*|_B$. 
    Consider any biconnected component of $N^*|_B$, say $C$. 
    Select any vertex in $C$, say $z$. 
    Since $B$ was biconnected, 
    it contained a pair of disjoint simple paths from $v$ to $z$. 
    Since $v$ had degree $3$, the paths must pass through two of $w$, $x$ and $y$. 
    Deleting $v$ removes only the initial edge of each path,
    so each path remains a simple path from one of $w$, $x$ and $y$ to $z$. 
    Since the original paths were internally vertex-disjoint, 
    these shortened paths remain internally vertex-disjoint. 
    Adjoin to each path at the beginning the adjacent leaf, 
    namely the leaf adjacent to $w$ if the path starts at $w$, 
    $\ell_1$ if it starts at $x$ and $\ell_2$ if it starts at $y$. 
    We now have a pair of disjoint paths from a leaf to $z$, 
    which we shall refer to as $P_1$ and $P_2$.

    Consider the last cut-edge on each of $P_1$ and $P_2$; 
    such must exist, as each path begins with a leaf edge, which is necessarily a cut-edge. 
    Every edge on the path after this final cut-edge is contained in 
    the same biconnected component as $z$, namely $C$, 
    and so $C$ has at least two incident cut-edges and is hence not a sink component. 
    As $C$ was selected arbitrarily, 
    $N^*|_B$ contains no sink component, 
    and $N^*$ is robustly orientable. 

    We now construct an orientation that is not tree-child. 
    Choose an edge $e$ that is not one of the new edges or the leaf edge $\{w,\ell\}$. 
    Since $N^*$ is robustly orientable, there is a phylogenetic orientation $o(N^*,e)$. 
    Observe in particular, that as $o(N^*,e)$ is phylogenetic, 
    $\ell_1$ and $\ell_2$ are sinks and so have in-degree one, and similarly for $\ell$. 
    Indeed, as $w$ has degree $2$ and is not a root or a leaf, 
    it must have in-degree and out-degree $1$. 
    An example orientation is depicted in Figure~\ref{fig:tc.construction} (c).

    We now apply the construction from Lemma~\ref{l:leaf.fusion} to obtain 
    a new network $N^{**}$ with phylogenetic orientation $o^{*}$, 
    where $\ell_1$ and $\ell_2$ are identified and a new leaf $\ell_3$ is added adjacent to $\ell_1$, 
    so that $\ell_1$ has in-degree-$2$ and out-degree one, and so is a reticulation vertex. 
    This is depicted in Figure~\ref{fig:tc.construction} (d). 

    Finally, we identify $\ell_3$ and $w$ to obtain a network isomorphic to $N$ 
    (which we will unambiguously call $N$), 
    with inherited orientation $o^{**}$. 
    This is depicted in Figure~\ref{fig:tc.construction}~(e).
    
    In particular, $w$ now has in-degree $2$ and out-degree $1$ and is thus a reticulation vertex. 
    Certainly, $o^{**}$ has at least one leaf, 
    all leaves are sinks, and there is still only one root. 
    As $o^{*}$ was acyclic, any cycle in $o^{**}$ must pass through $(\ell_1,w)$. 
    But any path through this edge must either terminate at $w$ or $\ell$, so it cannot form a cycle. 
    Thus, $o^{**}$ is acyclic and thus a phylogenetic orientation. 

    Finally, we observe that $\ell_1$ and $w$ are both reticulation vertices, and hence form a stack, 
    so $o^{**}$ is a phylogenetic orientation that is not tree-child or stack-free. 
    Hence, $N$ is not robustly tree-child or robustly stack-free. 
\end{proof}

\begin{Theorem}\label{t:robust.equivalence}
    Let $N$ be a robustly orientable network. Then the following are equivalent:
    \begin{enumerate}
        \item $N$ is robustly tree-child;
        \item $N$ is robustly stack-free;
        \item $N$ is level-$k$ for some $k \le 1$.
    \end{enumerate}
\end{Theorem}

\begin{proof}
    This theorem follows from Theorems~\ref{t:yes.level1} and~\ref{t:not.level2}.
\end{proof}

We now turn our attention to robustly normal networks, 
which we show coincide exactly with the class of phylogenetic trees, 
by explicitly constructing a phylogenetic orientation that is not normal.

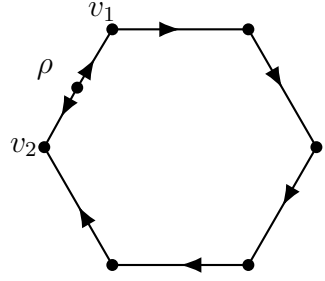
\begin{figure}
    \centering
\begin{tikzpicture}[scale=0.9]
    %\draw[help lines] (-5,-5) grid (5,5);
    \foreach \angle [count = \i] in {0, 60, ..., 300} {
        \fill (\angle:2cm) circle (2.5pt) coordinate (dot\i);
        %\node at (\angle: 2.4) {Dot \i};
    }

    \node at (150:2.3cm) {$\rho$};
    \fill (150:1.75cm) circle (2.5pt) coordinate (root);
    \node at (120:2.3cm) {$v_1$};
    \node at (180:2.3cm) {$v_2$};

    \draw [midarrow, thick]  (dot2) -- (dot1);
    \draw [midarrow, thick]  (dot3) -- (dot2);
    \draw [midarrow, thick]  (root) -- (dot3);
    \draw [midarrow, thick]  (root) -- (dot4);
    \draw [midarrow, thick]  (dot1) -- (dot6);
    \draw [midarrow, thick]  (dot6) -- (dot5);
    \draw [midarrow, thick]  (dot5) -- (dot4);
\end{tikzpicture}
    \caption{An example of our construction for normal networks in Theorem~\ref{t:robust.normal}. 
    The vertex $v_2$ is a reticulation vertex, and the edge $(\rho,v_2)$ is a shortcut.}
    \label{fig:normal.construction}
\end{figure}

\begin{Theorem}\label{t:robust.normal}
    Let $N$ be a robustly orientable phylogenetic network. 
    Then $N$ is robustly normal if and only if $N$ is an unrooted phylogenetic tree.
\end{Theorem}
 
\begin{proof}
    Suppose that $N$ is robustly normal. 
    Then $N$ is necessarily robustly tree-child, as every normal network is tree-child.
    Therefore by Theorem~\ref{t:robust.equivalence}, $N$ is level-$k$ for some $k \le 1$.

    Suppose for the sake of contradiction that $N$ is a level-$1$ network. 
    Consequently, there exists some blob in $N$ containing an undirected simple cycle $C$.  

    Let $e = \{v_1,v_2\}$ be an edge in $C$. 
    As $N$ is level-$1$, 
    there exists a phylogenetic orientation $o(N,e)$ by Theorem~\ref{t:level2robust}. 
    Consider any such phylogenetic orientation $o(N,e)$. 
    After subdividing the edge $e$ with the new root vertex $\rho$, 
    any valid phylogenetic orientation $o(N,e)$ must have directed edges $(\rho,v_1)$ and $(\rho,v_2)$.

    We construct a new orientation $o^*(N,e)$ as follows. 
    The orientations of all edges outside of $C$ are retained from $o(N,e)$, 
    as are directed edges $(\rho,v_1)$ and $(\rho,v_2)$. 
    For the remaining edges of $C$, 
    we orient them so that they form a single directed path $P$ along $C$ from $v_1$ to $v_2$. 
    We depict an example of this construction, restricted to a single biconnected component, 
    in Figure~\ref{fig:normal.construction}.

    First, we verify that $o^*(N, e)$ remains a valid phylogenetic orientation by confirming that 
    it is acyclic and preserves the root and leaf sets of $o(N,e)$. 
    As $o(N, e)$ is a valid orientation with its root in $C$, all directed paths must originate from $C$. 
    To preserve reachability, all cut-edges incident with $C$ must be oriented outward. 
    Since these outward orientations are retained in $o^*(N, e)$ 
    and every non-root vertex in $C$ is adjacent to a cut-edge, 
    every vertex in $C$ maintains an out-degree of at least $1$ under $o^*(N, e)$, 
    ensuring that no new leaves are created. 
    Furthermore, every vertex on $C$ except $\rho$ has an in-degree of at least $1$: 
    $v_1$ receives the edge $(\rho, v_1)$, 
    and all other vertices receive an incoming edge from the directed path~$P$. 
    Thus, no new roots are created and $o^*(N, e)$ has precisely the same leaf set and roots as $o(N, e)$.

    To see that $o^*(N, e)$ is acyclic, 
    note that all edges outside of $C$ retain their acyclic orientation from $o(N, e)$. 
    Consequently, any newly formed directed cycle must contain at least one edge from $C$. 
    However, because all cut-edges incident to $C$ are oriented away from the cycle, 
    any cycle must be entirely contained in $C$. 
    In particular, by construction no directed cycles are contained in $C$, so $o^*(N,e)$ is acyclic.
    
    Finally, observe the local structure of $o^*(N, e)$ at $v_2$. 
    There is a directed path $P$ via $v_1$ to $v_2$, 
    which is initiated by $(\rho,v_1)$ as well as a directed edge $(\rho,v_2)$. 
    This configuration forms a shortcut and so, although $o^*(N, e)$ is a phylogenetic orientation, 
    it is not normal, forming a contradiction. 
    Therefore, if $N$ is robustly normal, then it must be level-$0$ and hence a phylogenetic tree. 

    Conversely, suppose that $N$ is an unrooted phylogenetic tree. 
    By Theorem~\ref{t:level1.structure}, $N$ is robustly orientable. 
    Consequently, for any edge $e$ in $N$, there exists a phylogenetic orientation $o(N,e)$. 
    Under this orientation, $N$ is a directed phylogenetic tree and so must be a normal phylogenetic network.  
\end{proof}

\section*{Discussion}

In the present manuscript, 
we have completely characterised the class of robustly orientable networks as 
precisely those networks with no sink components 
(in the spirit of the classification of orientable networks in~\cite{bulteau2023turning}), 
and showed that it contains the classes of level-$k$ networks for $k \le 2$, and tree-based networks. 
Furthermore, we demonstrated that an unrooted network is robustly stack-free or robustly tree-child 
if and only if it is level-$1$ or less, and robustly normal if and only if it is level-$0$, 
that is, a phylogenetic tree.

Summarising, robust orientability is a surprisingly broad class, 
containing tree-based and all level-$k$ networks for $k \le 2$, 
while robust membership of more restrictive rooted classes collapses to very simple unrooted structures.

We expect future research in this context to fall mainly in three directions. 
Firstly, many of our proofs required the construction of certain forbidden directed substructures. 
There are several remaining classes of directed phylogenetic networks 
that can also be characterised by forbidden substructures, such as
fully tree-sibling networks~\cite{francis2026counting}, 
tree-based networks~\cite{francis2015phylogenetic} and 
labellable networks~\cite{francis2023labellable}. 
Characterising additional classes of robustly class-$\mathcal{C}$ networks is 
therefore one direction of possible future research.

Another major direction is to extend our results to non-binary networks. 
Several of our results take advantage of the fact that the networks in question are binary.
It would be interesting to see how these results could be extended or modified to the non-binary context.

Finally, one can consider problems in which partial information is given. 
In the present manuscript, 
we considered the problem of finding networks 
where each choice of root-edge and phylogenetic orientation yields an oriented network that belongs to some class. 
One could also consider the problem of finding network-edge tuples $(N,e)$ which are robustly class $\mathcal{C}$,
representing a circumstance where the root location is known but not necessarily the orientation of every edge, 
or problems in which the locations of reticulation vertices are known but the root is not, 
and determining whether every phylogenetic orientation that respects the reticulation vertex locations 
are members of some class $\mathcal{C}$.

\section*{Acknowledgements} 

The authors wish to thank Simone Linz for helpful initial conversations at the inception of this project, 
at the International Workshop on Probabilistic, Combinatorial, and Algorithmic Phylogenetics, Taipei, 2026. 
The authors also wish to thank Andrew Francis and Sky Basire for helpful discussions 
throughout the development of this manuscript. 
MH's research was supported by the Australian Government through 
the Australian Research Council's Discovery Projects funding scheme (Project DP260102678). 
The views expressed herein are those of the authors 
and are not necessarily those of the Australian Government or Australian Research Council.

\section*{Conflict of interest} 

The authors herewith certify that they have no affiliations with or involvement in any
organisation or entity with any financial 
(such as honoraria; educational grants; participation in speakers’ bureaus;
membership, employment, consultancies, stock ownership, or other equity interest; 
and expert testimony or patent-licensing arrangements) 
or non-financial (such as personal or professional relationships, affiliations, knowledge or beliefs) 
interest in the subject matter discussed in this manuscript.

\section*{Data availability statement} 

Data sharing is not applicable to this article as no new data were created or analysed in this study.

\bibliographystyle{plainnat}
\bibliography{sample}

\end{document}